\documentclass{article}

\usepackage{graphicx,amsthm,amssymb,amsmath,authblk}
\usepackage{hyperref,yhmath,physics,dsfont}

\newtheorem{theorem}{Theorem}

\title{Additive codes from linear codes 
\thanks{This research is supported by the Spanish Ministry of Science, Innovation and Universities grant PID2023-147202NB-I00 funded by .}}

\author[1]{Simeon Ball\thanks{Email: simeon.michael.ball@upc.edu.}}
\author[2]{ Tabriz Popatia\thanks{Email: tabriz.popatia@upc.edu. }}

\affil[1]{Dept.~of Mathematics, Universitat Politecnica Catalunya, 08034 Barcelona}
\affil[2]{Dept.~of Mathematics, Universitat Politecnica Catalunya, 08034 Barcelona}

\begin{document}

\maketitle

\begin{abstract}
We introduce two constructions of additive codes over finite fields. Both constructions start with a linear code over a field with $q$ elements and give additive codes over the field with $q^h$ elements whose minimum distance is demonstrably good.
\end{abstract}

\section{Introduction}

Additive codes have become of increasing importance in the field of quantum error-correction due to their equivalence to subgroups of the Pauli group \cite{KKKS2006} and also in the field of classical error-correction, as they can provide examples of codes which outperform linear codes  \cite{GLLM2023, Kurz2024}. It is perhaps surprising that additive codes have not been more widely studied until recently. The authors of this article, recently proved some Griesmer type bounds for additive codes, see \cite{BLP2025}.

In this article we will detail two general constructions of additive codes, both inspired by the constructions in \cite{GLLM2023}. The article is structured as follows. In Section~\ref{ghw}, we recall the definition of generalised Hamming weight of a linear code. In Section~\ref{con1section}, we provide a general construction of additive codes from linear codes, generalising Theorem 1 from \cite{GLLM2023}, which constructs additive codes form linear {\em cyclic} codes. The minimum distance of the additive code is bounded below by a generalised Hamming weight which itself can be bounded by the Griesmer bound.

In Section~\ref{con2section}, we provide a construction of additive codes which have a small dimension. The example, Example 3 from \cite{GLLM2023} lies within this family of additive codes. 
In Section~\ref{griesmersection}, we prove that in certain cases the construction in Theorem~\ref{construct2}, attains the additive Griesmer bound from \cite{BLP2025}.
In Section~\ref{performsection}, we discuss the known additive codes which outperform linear codes and consider future research in this area.

\section{Generalised Hamming weights} \label{ghw}

Let ${\mathbb F}_q$ denote the finite field with $q$ elements and let PG$(k-1,q)$ denote the $(k-1)$-dimensional projective space over ${\mathbb F}_q$.

We use the notation $[n,k/h,d]_q^h$ code to denote an additive code over ${\mathbb F}_{q^h}$ of size $q^k$, which is linear over ${\mathbb F}_q$ and of minimum distance $d$. Thus, a  $[n,k/h,d]_q^1$ code is a linear code.

The weight of a vector $v \in {\mathbb F}_q^n$ is the number of non-zero coordinates that it has. Let $C$ be an $[n,k,d]_q$ linear code, that is a $k$-dimensional subspace of ${\mathbb F}_q^n$ in which all non-zero vectors have weight at least $d$ and some vector has weight $d$. 

For any sub-code (subspace) $D$ of $C$, let
$\mathrm{Support}(D)$ be the set of coordinates $i$ for which there exists a $v \in D$ with $v_i\neq 0$.
The $j$-th generalised Hamming weight of $C$ is
$$
d_j(C)=\{ \min |\mathrm{Support}(D) | \ | \ D \leqslant C, \ \dim D =j \}.
$$
Observe that $d_1=d$, the minimum distance of $C$.

Let $G$ be a $k \times n$ generator matrix for $C$, ie. the row space of $G$ is $C$. Let $\mathcal X$ be the set of columns of $G$ viewed as points of PG$(k-1,q)$. 

A codeword (a $1$-dimensional subcode) of $C$ is obtained by left multiply $G$ be a vector $v \in {\mathbb F}_q^k$. Considering $v$ as a point of PG$(k-1,q)$, the hyperplane $v^{\perp}$ contains the point $x$ if and only if
$$
v_1x_1+\cdots+v_kx_k=0.
$$
Thus, the codeword $vG$ has weight $w$ precisely when $n-w$ points of $\mathcal X$ are contained in the hyperplane $v^{\perp}$. More generally, if we take a $j$-dimensional subspace $V$ of ${\mathbb F}_q^k$ then the co-dimension $j$ subspace $V^{\perp}$ contains $n-w$ points of $\mathcal X$ precisely when the subcode
$$
D(V)=\langle vG \ | \ v \in V \rangle
$$
has $|\mathrm{Support}(D(V))|=w$. Thus,
$$
|\mathrm{Support}(D(V))|=n-|V^{\perp} \cap \mathcal X|.
$$
and the $j$-th generalised Hamming weight of $C$,
\begin{equation} \label{genweighteqn}
d_j(C)= n- \max |\{ |V^{\perp} \cap \mathcal X| \ | \ \dim V=j,\ V\leqslant {\mathbb F}_q^k \}|.
\end{equation}

\section{A construction of additive codes from a linear code and a partial semifield} \label{con1section}

In the next theorem we use a set $\mathcal A=\{ A_1,\ldots,A_h\}$ of $k \times k$ matrices over ${\mathbb F}_q$ with the property that any linear combination
$$
 \sum_{j=1}^h \lambda_j A_j
$$
is non-singular. If $k=h$ then such a set of matrices is equivalent to a semifield. For this reason we call such a set with $h \leqslant k$ a {\em partial semifield}. For more on semifields and their applications to codes, see \cite{Lavruaw2012} and \cite{Sheekey2019}.

Let $\{e_1,\ldots,e_h\}$ be a basis for the field ${\mathbb F}_{q^h}$ over ${\mathbb F}_q$.

\begin{theorem} \label{partialsemiconstruct}
Let $C$ be a linear $[n,k,d]_q$, let $\mathcal A$ be a partial semifield and let $G$ a generator matrix for $C$ whose $j$-th column is $x_j$. Then the ${\mathbb F}_q$-subspace $C_{\mathrm{add}}$ spanned by the rows of the matrix
$G_{\mathrm{add}}$, whose $i$-th column is
$$
\sum_{j=1}^h A_j x_i e_j,
$$
is a $[n,k/h,\geqslant d_h(C)]_q^h$ additive code. 
\end{theorem}

\begin{proof}
Let 
$$
\mathcal X=\{x_1,\ldots,x_n\},
$$
where we view $x_i$ as a point of PG$(k-1,q)$.

For any $v \in {\mathbb F}_q^k$, the vector $vG_{\mathrm{add}} \in {\mathbb F}_{q^h}^n$ is a codeword of $C_{\mathrm{add}}$. The $i$-th coordinate of $v$ is zero if and only if
\begin{equation} \label{eqnsum}
0=\sum_{j=1}^h v A_j x_i e_j
\end{equation}
which is if and only if
$$
0=v A_j x_i
$$
for all $j \in \{1,\ldots,h\}$, which is if and only if the point $x_j$ is in the perpendicular space to
$$
V=\langle v A_1,\ldots,vA_h \rangle.
$$
Since $\mathcal A$ is a partial semified, the subspace $V$ has rank $h$, so the perpendicular space is a co-dimension $h$ subspace. 

Therefore, the weight of the codeword $vG_{\mathrm{add}}$ is
$$
n-|V^{\perp} \cap \mathcal X|
$$
which implies by (\ref{genweighteqn}) that the minimum distance is bounded below by $d_h(C)$.
\end{proof}

Observe that taking a different generator matrix $BG$ for the code $C$, where $B$ is a non-singular $k\times k$ matrix, may give a code with different parameters. Equation (\ref{eqnsum}) becomes
$$
0=\sum_{j=1}^h v A_j B x_i e_j
$$
which is equivalent to changing the co-dimension $h$ subspace defined $V$ to the subspace
$$
\langle v A_1B,\ldots,vA_hB \rangle.
$$
Obviously, this can have different containment distribution with respect to the set $\mathcal X$.

Theorem~\ref{partialsemiconstruct} has the following corollary.

\begin{theorem}
If there is a linear $[n,k,d]_q$ code then there is a $[n,k/h,d_{\mathrm{add}}]_q^h$ additive code where
$$
d_{\mathrm{add}} \geqslant \sum_{j=0}^{h-1} \lceil \frac{d}{q^j} \rceil.
$$
\end{theorem}

\begin{proof}
Consider a $j$-dimensional subcode $D$ of a  linear $[n,k,d]_q$ code $C$. Since $D$ is a $j$-dimensional code of minimum distance at least $d$, the Griesmer bound implies
$$
|\mathrm{Support}(D)|  \geqslant \sum_{j=0}^{j-1} \lceil \frac{d}{q^j} \rceil.
$$
and so
$$
d_j(C)\geqslant \sum_{j=0}^{j-1} \lceil \frac{d}{q^j} \rceil.
$$
Applying Theorem~\ref{partialsemiconstruct} completes the proof.
\end{proof}

These theorems extend Theorem 1 from \cite{GLLM2023} which would be the case that $C$ is a cyclic code and $h=2$. The example that they provide after Theorem 1, Example 3, does not actually satisfy the hypothesis of  \cite[Theorem 1]{GLLM2023}, since 
$$
\frac{x^{63}-1}{g}=x^{10}+x^9+x^6+x^4+x^3+x^2+x+1
$$
and $f_0+f_1$ have $x+1$ as a common factor. However, the additive $[63,5,45]_2^2$ code does have the parameters that the authors claim. This led us to discover the construction in Theorem~\ref{construct2}, which we will prove in the next section.

\section{A construction of additive codes of small dimension}  \label{con2section}

In the following theorem, we will construct a set $\mathcal X$ of rank $h$ (projective $h-1$ dimensional) subspaces of PG$(k-1,q)$ with the property that any hyperplane contains at most $n-d$ subspaces of $\mathcal X$. Using $\mathcal X$, we make a $k \times n$ generator matrix $\mathrm G$ for an additive $[n,k/h,d]_q^h$ code $C$ which has as $j$-th column 
$$
e_1x_1+\cdots+e_hx_h
$$
where 
$$
\langle x_1,\ldots,x_h \rangle
$$
is the $j$-th element of $\mathcal X$ (assuming some ordering of the elements of $\mathcal X$).

Observe that a codeword is $a\mathrm G$ for some $a \in {\mathbb F}_q^k$ and that the $j$-th coordinate of $a\mathrm G$ is zero if and only if the $j$-th element of $\mathcal X$ is contained in the hyperplane 
$$
a \cdot (X_1,\ldots,X_k)=0.
$$


\begin{theorem} \label{construct2}
If $h\leqslant s$ and $t \geqslant 2$ then there is a $[q^{st}-1,\frac{st+s+1}{h},d]_q^h$ additive code where
$$
d \geqslant q^{st}-1-\frac{q^{st}-1}{q^s-1}q^{s-h}.
$$
\end{theorem}

\begin{proof}
Let $\{\lambda_1,\ldots,\lambda_h\}$ be non-zero elements of ${\mathbb F}_{q^s}$ and consider 
the $h \times h$ matrix $\Lambda=(\lambda_j^{q^{s-1-i}})$, where $i \in \{0,\ldots,h-1\}$. 

If the matrix $\Lambda$ is non-singular then there are $\alpha_j \in {\mathbb F}_{q^s}$ such that
$$
\sum_{i=0}^{h-1} \alpha_i  \lambda_j^{q^{i}}=0.
$$
The zeros of the equation
$$
\sum_{i=0}^{h-1} \alpha_i  x^{q^{i}}=0,
$$
viewing the elements of ${\mathbb F}_{q^s}$ as points of $AG(s,q)$, form a $(h-1)$-dimensional affine subspace containing the zero vector. Thus, if we choose $\{\lambda_1,\ldots,\lambda_h\}$ so that they span a $(h-1)$-dimensional affine subspace which does not contain the zero vector then the matrix $\Lambda$ is non-singular.

Therefore, there exist $\theta_1, \ldots, \theta_h  \in {\mathbb F}_{q^s}$ such that
$$
\sum_{j=1}^h \theta_j \lambda_j^{q^{s-1-i}}=0
$$
for $i \in \{0,\ldots,h-2\}$ and
$$
\sum_{j=1}^h \theta_j \lambda_j^{q^{s-h}}\neq 0.
$$
We will use the symmetric bilinear form over ${\mathbb F}_q$ defined on ${\mathbb F}_q \times {\mathbb F}_{q^s} \times {\mathbb F}_{q^{st}}$ by
$$
(a_1,a_2,a_3) \cdot (x_1,x_2,x_3)=a_1x_1+\mathrm{tr}(a_2b_2)+\mathrm{Tr}(a_3b_3)
$$
where $\mathrm{tr}$ denote the trace function from ${\mathbb F}_{q^s}$ to ${\mathbb F}_q$ and $\mathrm{Tr}$ denote the trace function from ${\mathbb F}_{q^{st}}$ to ${\mathbb F}_q$.

As discussed above we will construct a set $\mathcal X$ of size $q^{st}-1$ of rank $h$ ${\mathbb F}_q$-subspaces of the vector space ${\mathbb F}_q \times {\mathbb F}_{q^s} \times {\mathbb F}_{q^{st}}$. 
Precisely, for each non-zero $x \in {\mathbb F}_{q^{st}}$, we take the ${\mathbb F}_q$-subspace
$$
\langle (1,\lambda_1 N(x), \lambda_1 x),\ldots,(1,\lambda_h N(x), \lambda_h x) \rangle,
$$
where $N(x)=x^{(q^{st}-1)/(q^s-1)}$ is the norm function from ${\mathbb F}_{q^{st}}$ to ${\mathbb F}_{q^{s}}$.

To prove a bound on the minimum distance it suffices to prove that for any non-zero 
$$
(a_1,a_2,a_3) \in {\mathbb F}_q \times {\mathbb F}_{q^s} \times {\mathbb F}_{q^{st}}
$$
the equation
$$
a_1+\mathrm{tr}(a_2 \lambda_j N(x))+\mathrm{Tr} (a_3 \lambda_j x)=0
$$
is zero, for all $j \in \{1,\ldots,h\}$, for a bounded number of $x \in {\mathbb F}_{q^{st}}$.

Writing out the trace functions and the norm function this is
$$
a_1+\sum_{i=0}^{s-1}(a_2 \lambda_j x^{(q^{st}-1)/(q^s-1)})^{q^{s-i-1}}+\sum_{m=0}^{st-1}(a_3 \lambda_j x)^{q^{st-1-m}}=0.
$$
Multiplying by $\theta_j$ and summing, we have
$$
\sum_{j=1}^h \theta_j \left( a_1+\sum_{i=h-1}^{s-1}(a_2 \lambda_j x^{(q^{st}-1)/(q^s-1)})^{q^{s-i-1}}+\sum_{m=h-1}^{st-1}(a_3 \lambda_j x)^{q^{st-1-m}}\right)=0,
$$
since $\lambda^{q^{st}}=\lambda^{q^s}$ and
$$
\sum_{j=1}^h \theta_j \lambda_j^{q^{s-1-i}}=0,
$$
for $i\in \{0,\ldots,h-2\}$. Since $t\geqslant 2$ and
$$
\sum_{j=1}^h \theta_j \lambda_j^{q^{s-h}} \neq 0,
$$
the equation has degree
$$
\frac{(q^{st}-1)}{q^s-1}q^{s-h}.
$$
Thus the hyperplane $(a_1,a_2,a_3)^{\perp}$ will contain at most this number of the subspaces of $\mathcal X$. Therefore, the codeword we obtain from $(a_1,a_2,a_3)$ has at most this number of zeros. This implies that the ${\mathbb F}_q$ rank of the code is $st+t+1$, since we do not obtain the zero codeword from any non-zero $(a_1,a_2,a_3)$, and that the minimum distance is bounded by
$$
d \geqslant q^{st}-1-\frac{q^{st}-1}{q^s-1}q^{s-h}.
$$
\end{proof}

\section{Comparisons with the additive Griesmer bound} \label{griesmersection}

In this section we prove that the parameters of the code constructed in Theorem~\ref{construct2}, when $t=2$ and $s=h$, are optimal. This we do by applying the additive Griesmer bound \cite[Theorem 9]{BLP2025}. Observe that for these parameters, Theorem~\ref{construct2} implies the existence of a  $[q^{2h}-1,\frac{3h+1}{h},d]_q^h$ additive code with 
$$
d \geqslant q^{2h}-q^h-1.
$$

\begin{theorem}
If there is a $[q^{2h}-1,\frac{3h+1}{h},d]_q^h$ additive code then
$$
d \leqslant q^{2h}-q^h-1.
$$
\end{theorem}

\begin{proof}

Suppose that there is a $[q^{2h}-1,\frac{3h+1}{h},q^{2h}-q^h]_q^h$ code.

Theorem 9 from \cite{BLP2025} states that if there is a $[n,r/h,d]_{q}^h$ additive code  then
\begin{equation} \label{addgr}
n \geqslant \lceil r/h \rceil +d-m+ \left\lceil  \frac{d}{f(q,m)} \right\rceil,
\end{equation}
where $r=(k-1)h+r_0$, $k=\lceil r/h \rceil $, $1 \leqslant r_0 \leqslant h$,
$$
f(q,m)=\frac{q^{(m-2)h+r_0}(q^h-1)}{q^{(m-2)h+r_0}-1}
$$
where $m$ is such that
$$
q^{(m-2)h+r_0} < d \leqslant q^{(m-1)h+r_0}.
$$

In the case of a $[q^{2h}-1,\frac{3h+1}{h},q^{2h}-q^h]_q^h$ code, we have that $m=3$, $k=4$ and $r_0=1$. 

Moreover,
$$
f(q,3)=\frac{q^{h+1}(q^h-1)}{q^{h+1}-1}
$$
and so
$$
\frac{d}{f(q,3)}=q^h-\frac{q^{2h}+q^{h+1}-q^h-1}{(q^h-1)q^{h+1}}
$$
and hence,
$$
\left\lceil \frac{d}{f(q,3)} \right\rceil=q^h
$$
Thus, (\ref{addgr}) implies,
$$
n \geqslant 4+q^{2h}-q^h-3+q^h
$$
which is a contradiction since $n=q^{2h}-1$.
\end{proof}

In the cases that $h<s$ or $t \geqslant 3$, the bound from  \cite[Theorem 9]{BLP2025} does not give a contradiction if we assume a larger $d$, so it is possible that an additive code with a larger minimum distance might exist.

\section{Additive codes versus linear codes} \label{performsection}

Before we address the issue of determining the known parameter sets for which additive codes outperform linear codes, we first determine that the additive codes we have constructed here are not equivalent to linear codes. Recall, from the previous section, that an additive $[n,k/h,d]_q^h$ code $C$
is equivalent to a set $\mathcal X$ of $n$ subspaces of PG$(k-1,q)$ of dimension at most $h-1$, with the property that any hyperplane contains at most $n-d$ subspaces of $\mathcal X$.

It was proven in \cite[Lemma 3.6]{AB2023} that $C$ is equivalent to a linear code if and only if $\mathcal X$ is contained in a Desarguesian spread of $(h-1)$-dimensional subspaces of PG$(k-1,q)$. The Desarguesian spread has the property that the span of any $m$ subspaces of the spread spans a subspace whose rank (vector dimension) is a multiple of $h$. Thus, we can check that the codes in Theorem~\ref{construct2}. Firstly since the dimension of a linear code is integral we must have that $s(t+1)=-1$ modulo $h$.

It is sufficient to prove that the codes are not equivalent to linear to take $m$ subspaces of $\mathcal X$ from the proof of Theorem~\ref{construct2} and prove that the ${\mathbb F}_q$ rank of of the matrix
$$
\left(\begin{array}{ccc|ccc|c|ccc}
1 & \ldots & 1 & 1 & \ldots & 1&\ldots & 1& \ldots & 1\\
\lambda_1 N(a_1)  & \ldots & \lambda_h N(a_1) & 
\lambda_1 N(a_2) & \ldots & \lambda_h N(a_2) & \ldots &
\lambda_1 N(a_m) & \ldots & \lambda_h N(a_m) \\
\lambda_1 a_1 & \ldots & \lambda_h a_1 & 
\lambda_1 a_2 & \ldots & \lambda_h a_2 & \ldots &
\lambda_1 a_m &   \ldots & \lambda_h a_m \\
\end{array} \right)
$$
is not a multiple of $h$. We have verified this by computer for all $s \leqslant 8$. 

In particular, we verified that the $[63,5,45]_2^2$ code from  \cite[Example 3]{GLLM2023} is not equivalent to a linear code. According to Grassl's table \cite{Grassl}, it is not known whether a linear  $[63,5,45]_4^1$ code exists.

\medskip

We now make a short list of parameters of codes for which we know an additive code exists and the corresponding linear code d
oes not exist.

\medskip

There is a $[21,3,18]_3^2$  additive code due to Mathon \cite{DDHM2002}.

\medskip

There are examples of $[n,4,d]_2^2$,  $[n,3,d]_3^2$,  $[n,3,d]_2^3$,  $[n,3,d]_4^2$, for $n-d$ sufficiently large due to Kurz \cite[Table 2 and Table 3]{Kurz2024} and $[n,5,d]_2^2$ additive codes \cite[Table 7]{Kurz2024} and $[n,3,d]_5^2$ additive codes \cite[Table 8]{Kurz2024}.

\medskip

We now make a short list of parameters of codes for which we know an additive code exists but where we do not know if a linear code with the same parameters exists.

\medskip

There are six examples in Guan, Li, Liu, Ma \cite[Table 2]{GLLM2023}, all of which are additive codes over ${\mathbb F}_4$.

\vspace{1cm}

   Simeon Ball\\
   Departament de Matem\`atiques, \\
Universitat Polit\`ecnica de Catalunya, \\
Modul C3, Campus Nord,\\
Carrer Jordi Girona 1-3,\\
08034 Barcelona, Spain \\
   {\tt simeon.michael.ball@upc.edu} \
   
   \bigskip
   
      Tabriz Popatia\\
   Departament de Matem\`atiques, \\
Universitat Polit\`ecnica de Catalunya, \\
Omega Building, Campus Nord,\\
Carrer Jordi Girona 1-3,\\
08034 Barcelona, Spain \\
   {\tt tabriz.popatia@upc.edu} \
   
\end{document}